\documentclass[onecolumn,amsmath,amssymb,aps,pra,notitlepage,11pt,tightenlines]{revtex4-1}

\usepackage{braket}
\usepackage{amsmath}
\usepackage{amsfonts}
\usepackage{amsthm}
\usepackage[margin=3cm]{geometry}
\usepackage{graphicx}
\usepackage{xcolor}
\usepackage{natbib}
\usepackage[page]{appendix}
\usepackage{mathtools}
\DeclarePairedDelimiter{\ceil}{\lceil}{\rceil}
\usepackage{tabularx}

\newcommand{\nth}{\textsuperscript{th}}
\newcommand{\numth}[1]{$#1$\textsuperscript{th}}

\newtheorem{definition}{Definition}
\newtheorem{problem}{Problem}

\newtheorem{lemma}{Lemma}

\begin{document}

\title{A quantum compiler for qudits of prime dimension greater than 3}
\author{Luke E. Heyfron}
\author{Earl Campbell}
\affiliation{Department of Physics and Astronomy, University of Sheffield, Sheffield, UK}

\begin{abstract}
Prevailing proposals for the first generation of quantum computers make use of 2-level systems, or \emph{qubits}, as the fundamental unit of quantum information.
However, recent innovations in quantum error correction and magic state distillation protocols demonstrate that there are advantages of using $d$-level quantum systems, known as \emph{qudits}, over the qubit analogues.
When designing a quantum architecture, it is crucial to consider protocols for compilation, the optimal conversion of high-level instructions used by programmers into low-level instructions interpreted by the machine. In this work, we present a general purpose automated compiler for multiqudit exact synthesis based on previous work on qubits that uses an algebraic representation of quantum circuits called phase polynomials.   We assume Clifford gates are low-cost and aim to minimise the number of $M$ gates in a Clifford+$M$ circuit, where $M$ is the qudit analog for the qubit $T$ or $\pi/8$ phase gate. A surprising result that showcases our compiler's capabilities is that we found a unitary implementation of the CCZ or Toffoli gate  that uses 4 $M$ gates, which compares to 7 $T$ gates for the qubit analogue.
\end{abstract}

\maketitle

\section{Introduction}

Despite its ubiquity in computing, the choice to use binary instead of ternary  or some other numeral system is almost arbitrary. From a purely information theoretic perspective, there is no reason to prefer bits over $d$-value anologues, known as \emph{dits}. In fact, successful experiments into 3-value logic were realised in the form of the \emph{Setun}, a ternary computer built in 1958 by Sergei Sobolev and Nikolay Brustentsov at Moscow State University~\cite{Brusentsov_2011}. The near universal adoption of binary can be explained from an engineering perspective in that it is much simpler to manufacture binary components.  However, since as early as the 1940's with the biquinary Collossus computer, it has been widely understood that there are intrinsic efficiency benefits of using higher dimensional logic components in that fewer are required.
 
In the standard paradigm, there are three components required for a fault tolerant quantum computing architecture: quantum error correction (QEC) codes; magic state distillation (MSD) protocols; and finally, quantum compilers. For qudits, there has been progress showing that both qudit QEC~\cite{Duclos-Cianci_2013, Anwar_2014, Hutter_2015,Watson_2015_a,Watson_2015_b} and qudit MSD~\cite{anwar2012qutrit,Campbell_2012,Campbell_2014,haah2017magic,krishna2018towards} offer a resource advantage in shifting from qubits to qudits.  However, surprisingly little work has been done on qudit compiling, except for the special case of qutrits where $d=3$~\cite{khan2005synthesis,bocharov2017factoring}.    Therefore,  compiling is the crucial missing piece in understanding quantum computing with qudit logic beyond $d=3$.

A standard metric for quantum compilers to minimize is the number of expensive gates that require magic state distillation.
In the qubit case, the $T$ gate is typically the designated magic gate in the low-level instruction set and much progress has been made on gate synthesis in this context.
For single qubits, the Matsumoto-Amano normal form~\cite{Matsumoto_2008, Giles_2013, Kliuchnikov_2013} leads to decompositions of single qubit unitaries as sequences of gates from the Clifford + $T$ gate set that is optimal with regards to $T$ count for a given approximation error.
So for single qubits, the problem is essentially ``solved''.
For multi-qubit operators, methods for $T$-optimal exact compilation have been developed but suffer exponential runtime~\cite{Gosset_2014}.
More recently, efficient optimizers have been developed that successfully reduce $T$ count, some of which are based on a correspondence between unitaries on a restricted gate set and so called phase-polynomials~\cite{Amy_2014,Amy_2016,Campbell_2017,Heyfron_2019}, and others that are based on local rewrite rules~\cite{Nam_2017}.
For qudits, there has been some work on single qutrit (three level systems) synthesis~\cite{Glaudell_2018} that can be considered a qutrit generalisation of the Matsumoto-Amano normal form.

In this work, we borrow ideas from the phase-polynomial style $T$ count optimization protocols and apply these insights to qudits.
We provide a general purpose compiler for exact synthesis of multiqudit unitaries generated by $M$, $P_l$ and $SUM$ gates where the $M$ gate is the canonical ``expensive'' magic gate (i.e. the qudit analogue of the $T$ gate).
We present an example of a $M$ count reduction only possible for odd prime $d>3$.
This is the CCZ gate, which is known to have optimal $T$ count of 7 when synthesised unitarily using qubit based quantum computers, whereas our decomposition has $M$ count of 4.
Until now, this reduced cost has only been achieved in the qubit setting using non-unitary gadgets that exploit ancillas~\cite{Jones_2013}.

\begin{figure*}[t]
\centering \includegraphics[width=\linewidth]{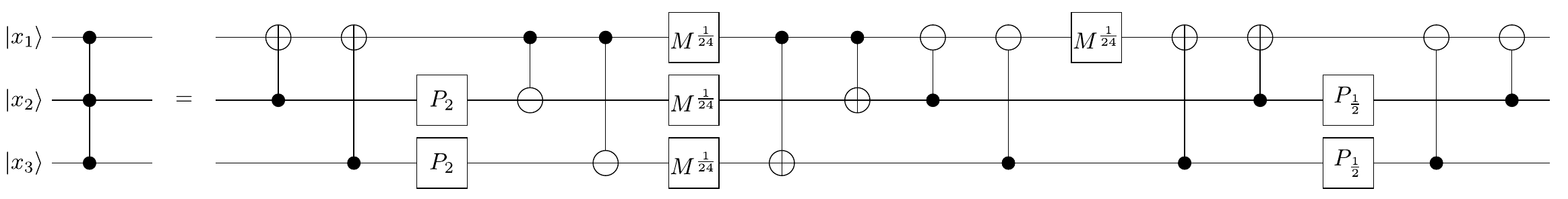}
\caption{A unitary implementation of the $CCZ$ gate that uses 4 non-Clifford $M^k$ gates. Note that the $\frac{1}{24}$ in the exponent of the $M^k$ gates is the multiplicative inverse of $24$ in the field $\mathbb{Z}_d$.  Gates used here are defined in Sec.~\ref{s_prelim}.}
\label{fig_ccz}
\end{figure*}

\section{Preliminaries}
\label{s_prelim}

Let $d>3$ be a prime integer. We define a Hilbert space on $n$ qudits spanned by the computational basis vectors $\{\ket{\mathbf{x}} \ \mid \ \forall \mathbf{x} \in \mathbb{Z}_d^n\}$.
We define the single qudit Pauli operators $X := \sum_{x \in \mathbb{Z}_d}\ket{x+1}\bra{x}$, where addition is performed modulo $d$; $Z := \sum_{x \in \mathbb{Z}_d}\omega^x\ket{x}\bra{x}$, where $\omega := \exp^{i\frac{2\pi}{d}}$ is a primitive $d$\nth root of unity.
The set of all $n$ qudit unitaries generated by $X$ and $Z$ form the Pauli group, $\mathcal{P}$.

The Clifford group, $\mathcal{C}$, is the normalizer of the Pauli group and is generated by: 
\begin{align}
H & := \frac{1}{\sqrt{d}}\sum_{x,y \in \mathbb{Z}_d} \omega^{xy}\ket{y}\bra{x} \\
SUM & :=  \sum_{c,t \in \mathbb{Z}_d}\ket{c,t+c}\bra{c,t}  \\
S & := \sum_{x \in \mathbb{Z}_d} \omega^{x^2}\ket{x}\bra{x}
\end{align} 
We refer to $H$ as the Hadamard; 
the Hadamard gate, $SUM$ is the two-qudit SUM gate; and $S$ is the \emph{phase gate}. Note that the $Z$ and $S$ gates are both diagonal and correspond to linear and quadratic terms, respectively, appearing in the exponent of the phase.

We further define the Clifford unitaries 
\begin{equation}
\label{Pl_gate}
P_l :=  \sum_{x \in \mathbb{Z}_d}\ket{l x}\bra{x} 
\end{equation} 
for all integer $l \neq 0 $, which we call \emph{product operators} as they perform field multiplication between the input basis states and a non-zero field element, $l$. It can be shown that all product operators are in the Clifford group.

As in previous works~\cite{Campbell_2012,Howard_2012,Cui_2017}, we define the canonical non-Clifford gate to be 
\begin{equation}
   M := \sum_{x \in \mathbb{Z}_d}\omega^{x^3}\ket{\mathbf{x}}\bra{\mathbf{x}} ,
\end{equation}
which lies in the third level of the Clifford hierarchy and in standard fault tolerant architectures are much more costly than Clifford gates due to the need for MSD.

\section{The Compiling Problem}
A compiler converts high-level instructions into low-levels ones.
In this paper, we concern ourselves with high-level instructions that take the form of $n$-qudit unitaries which can be exactly synthesised by a discrete gate set, $\mathcal{G}$.
By low-level instructions, we specifically refer to quantum circuits, which are represented as \emph{netlists}, or time-ordered lists of gates taken from $\mathcal{G}$ where the qudits to which they apply (as well as any other gate parameters) are specified for each gate.
The unitary that a particular quantum circuit implements is simply the right-to-left matrix product of each gate in the netlist extracted in time-order.

\begin{problem}
\label{prob_comp}
(Compiling Problem). Given a unitary $U \in \mathcal{G}$, find a quantum circuit that implements $U$ with the lowest cost.
\end{problem}

Note that the compiling problem is ill-defined and depends on the definition of cost.
The most accurate metric of quantum circuit cost is the full space-time volume, which is the number of machine level operations multiplied by the number of physical qubits.
The calculation required to determine the full space-time volume is lengthy and is highly sensitive to the choice of architecture~\cite{fowler2012towards,o2017quantum,babbush2018encoding}.
The \emph{$M$ count}, or the number of $M^k$ gates in a quantum circuit, is an alternative cost metric that gives a good approximation to the full cost and can be easily read off compiler-level quantum circuits.
Using the $M$ count in problem~\ref{prob_comp}, we obtain a well defined compiling problem.

\begin{problem}
($M$-Minimization). Given a unitary $U \in \mathcal{G}$, find a quantum circuit that implements $U$ with the fewest $M^k$ gates.
\end{problem}

We choose our gate set to be $\mathcal{G} = \{ Z, S, M^k, P_l, SUM \}$ for all available choices of $k$ and $l$.
While we would ideally work with a universal gate set such as Clifford + $M$,
the compiling problem is known to be intractable in the universal case so we focus on this simpler sub-problem.
For the selection of gates in $\mathcal{G}$, we have taken inspiration from previous work~\cite{Amy_2014, Amy_2016, Campbell_2017, Heyfron_2019},
where it was demonstrated that such a restriction leads to an algebraic reformulation of the compiler problem that is more amenable to computational methods, including efficient heuristics.

\section{Phase Polynomial Formalism}
\label{s_phasepoly}

The formalism described in this section allows us to reframe the $M$-minimization problem as a computationally-friendly problem on integer matrices.
It applies strictly to unitaries \\ $U \in \langle Z, S, M^k, P_l, SUM \rangle$ and is a straightforward generalisation of previous work~\cite{Amy_2014, Amy_2016, Campbell_2017, Heyfron_2019}.
We proceed with a lemma that establishes a correspondence between unitaries generated by $\mathcal{G}$ and cubic polynomials that we call \emph{phase polynomials}.

\begin{lemma}
\label{thm_pf2cub}
Any $n$ qudit unitary $U_f \in \langle \mathcal{G} \rangle$ can be expressed as follows:
\begin{equation}
\label{eq_pf}
U_f = \sum_{\mathbf{x} \in \mathbb{Z}_d^n} \omega^{f(\mathbf{x})}\ket{E\mathbf{x}}\bra{\mathbf{x}},
\end{equation}
where $E$ is an invertible matrix implementable with $SUM$ gates, and $f: \mathbb{Z}_d^n \mapsto \mathbb{Z}_d$ is a polynomial of order less than or equal to 3.
\end{lemma}
\begin{proof}
To prove the first part, we first show that each gate in the generating set can be written in the above form, then show that the set generated by these operators form a group.
From the definitions provided in section~\ref{s_prelim}, we have that $Z$, $S$ and $M$ gate applied to the $t$\textsuperscript{th} qudit can be written in the form of equation~\eqref{eq_pf} with $f(\mathbf{x}) = x_t, x_t^2, x_t^3$, respectively, and with $E = \mathbb{I}$.
$P_l$ applied to the $t$\textsuperscript{th} qudit  has $f(\mathbf{x}) = 0$ (as does the $SUM$ gate) and $E=\mathbb{I} + ((l-1)\delta_{i,t}\delta_{j,t})$ with inverse $E^{-1}=\mathbb{I} + ((\frac{1}{l}-1)\delta_{i,t}\delta_{j,t})$.
Finally, the $SUM$ gate whose control and target are the $c$\textsuperscript{th} and $t$\textsuperscript{th} qudits, respectively, has $E=\mathbb{I} + (\delta_{i,t}\delta_{j,c})$,  which has inverse $E^{-1}=\mathbb{I} - (\delta_{i,t}\delta_{j,c})$. 
By definition, the set generated by $\mathcal{G}$ is closed under multiplication and as each generator is a unitary matrix, the associative property holds. Finally, $\mathbb{I}, Z^\dagger, S^\dagger, M^\dagger \in \langle \mathcal{G} \rangle$ so the identity and inverse group axioms are satisfied.

To prove the second part, that $f(\mathbf{x})$ is cubic, we note that the only gates which contribute to $f(\mathbf{x})$ are $Z$, $S$ and $M$, which add a term equal to the state of the acted-upon qudit raised to the first, second and third power, respectively.
Because the $Z$, $S$ and $M$ gates are diagonal, the state of any qudit at any point in the circuit can only change due to the $P_l$ and $SUM$ gates, which together map the state of each qudit to linear functions of the input states with coefficients in $\mathbb{Z}_d$.
The linear functions can, at most, be raised to the $3$\textsuperscript{rd} power (due to the $M$ gate), before contributing a term to $f(\mathbf{x})$. Therefore, the order of $f(\mathbf{x})$ is at most cubic.
\end{proof}

The linear and quadratic terms of any $f(\mathbf{x})$ can be implemented using just Clifford operations, which cost considerably less than the cubic terms that require $M$ gates.
Therefore, we assume that $f(\mathbf{x})$ is a homogeneous cubic polynomial. It follows that $f(\mathbf{x})$ can be decomposed in the monomial basis as follows:
\begin{equation}
\label{eq_mon}
f(\mathbf{x}) = \sum_{\alpha,\beta,\gamma=1}^nS_{\alpha,\beta,\gamma}x_\alpha x_\beta x_\gamma,
\end{equation}
where $S \in \mathbb{Z}_d^{(n,n,n)}$. 
Since every choice of $(\alpha,\beta,\gamma)$ for $\alpha \leq \beta \leq \gamma$ corresponds to a different linearly independent monomial, if we enforce that $S$ is symmetric, it follows that the elements of $S$ uniquely determine the function $f(\mathbf{x})$. For this reason, we call it the \emph{signature tensor}.

The phase polynomial $f(\mathbf{x})$ can also be decomposed as a sum over linear forms raised to the third power, as in the following:
\begin{equation}
\label{eq_imp}
f(\mathbf{x}) = \sum_{j=1}^{m}\lambda_j\left(\sum_{i=1}^{n}A_{i,j}x_i\right)^3,
\end{equation}
where $\mathbf{\lambda} \in (\mathbb{Z}_d \setminus \{0\})^m$ and $A \in \mathbb{Z}_d^{(n,m)}$ such that for each column in $A$, there is at least one non-zero element.
It is straightforward to calculate the signature tensor from the elements of $A$ and $\lambda$ using the following relation,
\begin{equation}
\label{eq_st}
S_{\alpha,\beta,\gamma} = \sum_{j=1}^m \lambda_j A_{\alpha,j} A_{\beta,j} A_{\gamma,j}.
\end{equation}

\begin{definition}
\noindent \textbf{Implementation. } Let $U_f$ be a unitary with signature tensor $S \in \mathbb{Z}_d^{(n,n,n)}$.
Let $A \in \mathbb{Z}_d^{(n,m)}$ and $\mathbf{\lambda} \in (\mathbb{Z}_d \setminus \{0\})^m$.
We say that the tuple $(A, \lambda)$ is an \emph{implementation} of $S$ if it satisfies equation~\eqref{eq_st}.
\end{definition}

We refer to the tuple $(A, \lambda)$ as an implementation because it reveals information sufficient to construct a quantum circuit that implements $U_f$ with known $M$ count, as stated in the following lemma.
\begin{lemma}
\label{lem_imp}
Let $U_f$ be a unitary with an implementation $(A , \lambda)$ that has $m$ columns.
It follows that a quantum circuit can be efficiently generated which implements $U_f$ using no more than $m$ $M$ gates.
\end{lemma}
As proof of lemma~\ref{lem_imp}, we provide in appendix~\ref{ap_proof} an explicit algorithm for efficiently converting an implementation with $m$ columns into a quantum circuit with $m$ $M^k$ gates.

The connection between column count of implementations and $M$ count of quantum circuits is central to the understanding of this work and leads to a restatement of the compiler problem that is more amenable to computational solvers.
\begin{problem}
\label{pr_col}
(Column-minimization). Let $S$ be a signature tensor.
Find an implementation $( A , \lambda )$ that implements $S$ with minimal columns.
\end{problem}

\section{Example: CCZ Gate}
\label{s_CCZ}

Take the $CCZ$ gate as an example, which acts upon the computational basis as follows.
\begin{equation}
CCZ\ket{x_1, x_2, x_3} = \omega^{x_1 x_2 x_3}\ket{x_1, x_2, x_3}.
\end{equation}
In the monomial basis, the phase polynomial can be read off directly as $f(\mathbf{x}) = x_1 x_2 x_3$, which corresponds to a signature tensor with $S_{\sigma(1,2,3)} = \frac{1}{6}$ for all permutations $\sigma$ and $S_{\alpha,\beta,\gamma} = 0$ for all other elements.
However, to generate a quantum circuit for $U_f=CCZ$, we first need to find an implementation for $S$.
By applying knowledge of the qubit version of the $CCZ$ gate to qudits~\cite{Amy_2014, Amy_2016, Campbell_2017,Heyfron_2019}, we arrive at the following implementation\footnote{Note that for notational convenience, we often write an implementation as a single matrix where $\lambda$ is the final row and the rest is the $A$ matrix with a separating horizontal line between them.} that has an $M$ count of 7:
\begin{equation}
\begin{pmatrix} A \\ \hline \lambda \end{pmatrix} = \begin{pmatrix}	1 & 0 & 0 & 1 & 1 & 0 & 1   \\
				0 & 1 & 0 & 1 & 0 & 1 & 1  \\
				0 & 0 & 1 & 0 & 1 & 1 & 1 \\
				\hline \frac{1}{6} & \frac{1}{6} & \frac{1}{6} & -\frac{1}{6} & -\frac{1}{6} & -\frac{1}{6} & \frac{1}{6}\end{pmatrix}, \\
\end{equation}
which corresponds to the phase polynomial,
\begin{equation}
\begin{split}
\label{eq_ccz_imp1}
f(\mathbf{x}) &= \frac{1}{6} x_1^3 + \frac{1}{6} x_2^3 + \frac{1}{6} x_3^3 - \frac{1}{6} (x_1 + x_2)^3 \\&- \frac{1}{6} (x_1 + x_3)^3 - \frac{1}{6} (x_2 + x_3)^3 + \frac{1}{6} (x_1 + x_2 + x_3)^3.
\end{split}
\end{equation}
We remind the reader that all elements of an implementation are in $\mathbb{Z}_d$, so the fraction $\frac{1}{6} = x \in \mathbb{Z}_d$ where $x$ solves $6x = 1 \pmod{d}$.
One can easily verify that the above implementation, $(A,\lambda)$, satisfies equation~\eqref{eq_st} for every element of the signature tensor, $S_{a,b,c}$, confirming that it implements the $CCZ$ gate.
 
Using a computer aided discovery method described in section~\ref{s_opt}, we have found an implementation with $M$ count 4 that works for all choices of $d$. This is a key result of the present work and is provided below.
\begin{equation}
\label{eq_ccz_imp}
\begin{pmatrix}  A \\ \hline \lambda \end{pmatrix} = \begin{pmatrix}	1 & 1 & -1 & -1    \\
				1 & -1 & 1 & -1 \\
				1 & -1 & -1 & 1 \\
			\hline \frac{1}{24} & \frac{1}{24} & \frac{1}{24} & \frac{1}{24} \end{pmatrix}.
\end{equation}
This corresponds to the phase polynomial
\begin{equation}
\begin{split}
\label{eq_ccz_imp2}
f(\mathbf{x}) &=  \frac{1}{24}(x_1 + x_2 + x_3)^3 + \frac{1}{24}(x_1 - x_2 - x_3)^3 \\&+\frac{1}{24} (x_2 - x_1 - x_3)^3 +\frac{1}{24} (x_3 - x_1 - x_2)^3.
\end{split}
\end{equation}
An explicit quantum circuit for the above implementation of  the $CCZ$ gate is provided in figure~\ref{fig_ccz}.

\section{Compilers}
\label{s_opt}

\subsection{Brute-Force}
In order to construct an $M$-optimal implementation for a given phase polynomial, one can perform a brute-force search over all possible implementations, checking in polynomial-time in each case that it corresponds to the correct signature tensor using equation~\eqref{eq_st}.
However, the size of the search space scales as $O(d^{(n+1)m})$, which makes execution times impractical, even for modest sized inputs.
However, by searching in $m$-order where $m$ is the candidate number of $M$ gates, we can optimally compile unitaries on $n=3$ ququints ($d=5$) with $M$ count of up to $4$ (and lower bound unitaries with an implicitly higher optimal $M$ count).
It was through this brute-force method that we were able to discover the implementation of $CCZ$ with $M$ count of 4 presented in equation~\eqref{eq_ccz_imp}.

\subsection{Monomial Substitution}
It is critically important that a general-purpose compiler is efficient. Fortunately, there is a simple method to map a phase polynomial in the monomial basis to an implementation.
There are three kinds of monomial that may appear in a phase polynomial, which are distinguished by the number of variables they take. These are $x_a^3$, $x_a x_b^2$ and $x_a x_b x_c$.
As each monomial is linearly independent, if we can find a prototypical implementation for each kind of monomial, then it follows that we can compile an implementation for a general phase polynomial by substituting instances of the prototypes for each monomial.
Again using~\cite{Heyfron_2019} as inspiration, we provide prototype implementations for the three kinds of monomial below.

\begin{align}
\label{eq_mon1}x_a^3 &\rightarrow x_a^3 \\
\label{eq_mon2} x_a x_b^2 &\rightarrow \frac{1}{6}(x_a + x_b)^3 + \frac{1}{6}(x_a - x_b)^3 
- \frac{1}{3}x_a^3 \\
 \label{eq_mon3} x_a x_b x_c &\rightarrow \frac{1}{24}(x_1 + x_2 + x_3)^3 +  \frac{1}{24} (x_1 - x_2 - x_3)^3 +  \frac{1}{24} (x_2 - x_1 - x_3)^3 +  \frac{1}{24}(x_3 - x_1- x_2)^3,
\end{align}
where we have used the implementation from equation~\eqref{eq_ccz_imp2} for the $x_a x_b x_c$ prototype.
Of course, we can also use the ``legacy'' $M$-count 7 implementation from equation~\eqref{eq_ccz_imp1}, which for certain input unitaries (e.g. ones that contain many gates on the same qudit lines) lead to lower $M$ count implementations due to column merging (see section~\ref{sec_opt}).

We call the above method \emph{monomial substitution}, which executes in time that scales as $O(n^3)$ in the worst case, making it efficient.
However, the output $M$ count should be considered a crude initial guess at the optimal $M$ count and can be significantly improved by the optimization methods described in remainder of this section.

\subsection{$M$-Optimization}
\label{sec_opt}
	One approach to solving problem~\ref{pr_col} is to try and `merge' columns of an existing implementation.
	A pair of columns can be merged if they are duplicates of one another.
	This is because we can collect like terms in the phase polynomial where the coefficients combine linearly.
	An illustrative example is the following. Let $f$ be a phase polynomial with two terms, hence has implementation matrix $A$ with two columns,
	\begin{align}
	f(\mathbf{x}) &= \lambda_1  (A_{1,1}x_1 + A_{2,1}x_2 + \dots + A_{n,1}x_n)^3 \\ &+ \lambda_2  (A_{1,2}x_1 + A_{2,2}x_2 + \dots + A_{n,2}x_n)^3
	\end{align}
	if the two columns of $A$ are duplicates, then we have $A_{i,1}=A_{i,2} \ \forall \ i \in [1,n]$. And so
	\begin{equation}
	f(\mathbf{x}) = (\lambda_1 + \lambda_2)  (A_{1,1}x_1 + A_{2,1}x_2 + \dots + A_{n,1}x_n)^3,
	\end{equation}
	which needs only a single column to represent it, and therefore only a single magic state to implement it.
	
	Of course, it is often the case that an $A$ matrix does not contain any duplicates.
	In this case, we wish to transform $A$ in some way in order to make it contain duplicates, and in such a way that it does not alter the unitary it implements.
	In appendix~\ref{ap_dam}, we describe an $M$-optimizer that systemically searches for and performs such ``duplication transformations'', and subsequently merges the duplicated columns.
	For this reason, we call it the Duplicate And Merge (DAM) algorithm.
	
	The algorithm runs in time that scales as $O(m n^3 d^m)$ so is inefficient.
	However, in practice, it executes much faster than the brute-force compiler and often outputs $M$-optimal implementations, albeit non-deterministically, and is useful for raising the practical limit on input circuits.
	
	\section{Benchmarks}
	In order to determine the speed benefits of using DAM over a brute force search (BFS) and to assess the inevitable drop in $M$-optimality, we performed a benchmark on randomly generated implementations with an $M$-count of 3 for $d=5$ and $n=3$.
	These parameters were chosen as they are the largest parameters that are feasible for BFS where many repetitions are required.
	Each of the $100$ random implementations were first compiled by BFS, then the legacy monomial substitution compiler was run using the signature tensor as input, which was subsequently optimized using DAM $1000$ times.
	The distribution of $M$-counts after optimization with DAM was recorded and an example of a single random instance is shown in figure~\ref{fig_hist}.
	
	\begin{figure}
	\centering \includegraphics[width=0.4\linewidth]{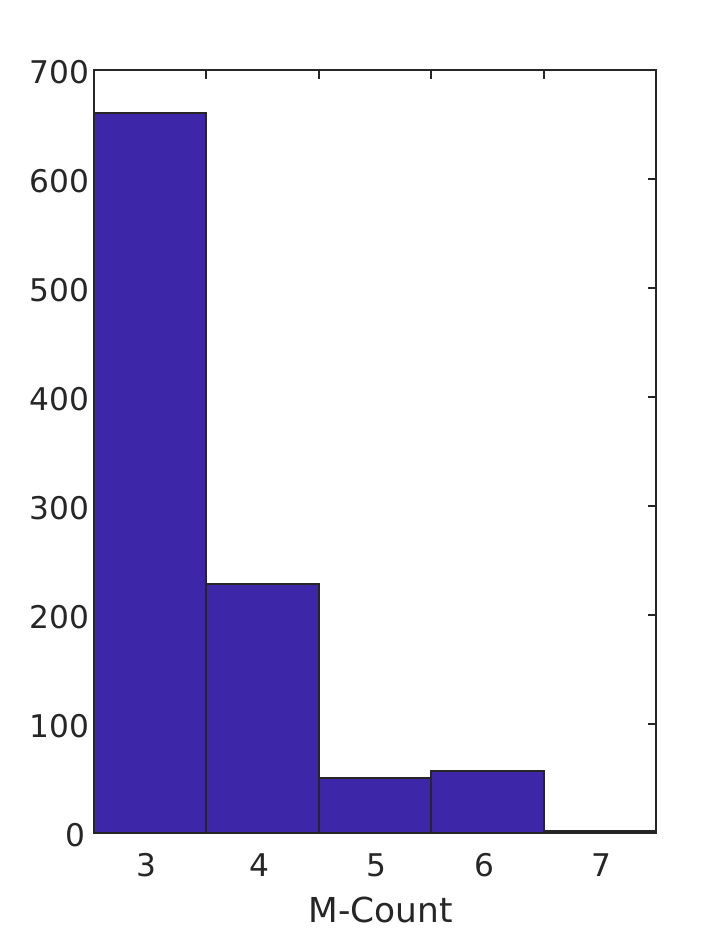}
	\caption{The distribution of $M$-counts for DAM run on a single random implementation with $d=5$, $n=3$ and known optimal $M$-count of 3 performed 1000 times.}
	\label{fig_hist}
	\end{figure}

We have performed a number of benchmarks on larger circuits in order to compare the monomial substitution, legacy monomial substitution and DAM compilers, which are shown in table~\ref{tab_1}. From the data we see that MS is the preferred compiler for low depth circuits such as the $CCZ^{\otimes k}$ family, which is to be expected as it is likely that the optimal $M$-count in this case is $4k$.
In contrast, the DAM compiler consistently outperforms MS and Legacy for random circuits.
Unfortunately, DAM is too inefficient to be practically useful for scalable quantum computers as  we see from the large execution times.
The computational bottleneck for DAM is due to the hardness of solving the multivariate cubic system of equations in equation~\eqref{eq_dodd1}.
Our implementation uses a search over the space of all vectors on $\mathbb{Z}_d$ that are of length $m$ in the worst case (for details see appendix~\ref{ap_dam}).
Therefore, the search consists of $O(d^m)$ iterations and as all other parts of the DAM algorithm executes in polynomial time, this is the sole source of inefficiency.
It follows that if one had access to an efficient heuristic that solves the system of equations in~\eqref{eq_dodd1}, then DAM immediately becomes efficient.
Although, it would probably be the case that as a consequence of it being a heuristic, not all column mergings would be discoverable.
	
	As DAM is a non-deterministic  heuristic, it is important to quantify the stability of its output.
	We used the results to estimate the probability that DAM produces an $M$-optimal implementation of an unknown quantum circuit for the given circuit parameters and found it to be $p_{\text{opt}}:=p(\text{optimal} \ | \ d=5, n=3, m=3) = 0.47\pm0.02$.
	This probability should in no way be interpreted as representative of DAM's performance in general (esp. not for larger circuits) but instead suggests that a best-of-$N$ approach (i.e. where DAM is run $N$ times and the implementation with the minimum $M$-count is returned) would be effective and may return $M$-optimal solutions.
	
	The probability of optimality can be used to estimate the minimum number of DAM repetitions, $N$, that one should perform in order to reach a desired confidence threshold, $p_{\text{conf}}$, using the following:
	\begin{equation}
	N = \ceil[\bigg]{\frac{\ln(1-p_{\text{conf}})}{\ln(1-p_{\text{opt}})}}.
	\end{equation}
	Using our results for $p_{\text{opt}}$ and a confidence threshold of $p_{\text{conf}}=0.95$, we arrive at $N=5$.
	The mean execution time\footnote{Execution times were obtained on a laptop with an Intel Core i7 2.40GHz processor, 12GB of RAM running Microsoft Windows 10 Home edition.} for a single run of DAM was $0.119s \pm 0.006s$ compared to $91s\pm6s$ for BFS.
	Therefore, the best-of-5 DAM compiler remains faster than BFS by two orders of magnitude.

	\begin{table}[h!]
	\caption{\small Benchmark table for the Monomial Substitution (MS), Legacy (Leg.) and Duplicate And Merge (DAM) compilers. 
	For each benchmark, a best-of-10 method was used for DAM. 
	The $M$-counts are shown in the $M_{\text{compiler}}$ columns and the executions times in the $t_{\text{compiler}}$ columns.
	The executions times were obtained on the Iceberg HPC Cluster at the University of Sheffield.
	The ``Random'' circuits are randomly generated  signature tensors where each element is a non-zero value with 50\% probability the subsequent value is selected uniformly.
	All $M$-counts and execution times reported in these rows are mean values taken over 100 random circuits, or as many circuits as could be synthesised within a 24 hour window. The notation CCZ$_{\# m}$ refers to a circuit with $m$ CCZ gates where they all share exactly 1 control qudit in common.}
	\label{tab_1}
		\begin{tabularx}{\textwidth}{|X|c|c|c|c|c|c|c|c|}
	\hline   \textbf{Circuit} & $\mathbf{d}$ & $\mathbf{n}$ & $\mathbf{M_{\text{Leg.}}}$ & $\mathbf{M_{\text{MS}}}$ & $\mathbf{M_{\text{DAM}}}$ & $\mathbf{t_{\text{Leg.}}}$ \textbf{(s)} & $\mathbf{t_{\text{MS}}}$  \textbf{(s)}& $\mathbf{t_{\text{DAM}}}$  \textbf{(s)} \\
	\hline   CCZ &5 & 3&7&4&5&0.10&0.02&0.98 \\
	CCZ$^{\otimes 2}$ &5 & 6 &14&8&10&0.03&0.01&20.87 \\
	CCZ$^{\otimes 3}$ &5 & 9 &21&12&16&0.02&$<$0.01&813.02 \\
	 CCZ &7 & 3 &7&4&7&0.08&0.03&6.99 \\
	CCZ$^{\otimes 2}$ &7 & 6 &14&8&10&0.04&0.01&182.69 \\
	CCZ &11 & 3 &7&4&7&0.09&0.04&38.81 \\
	CCZ$^{\otimes 2}$ &11 & 6 &14&8&12&0.04&0.01&11489.15 \\
	CCZ$_{\# 2}$ &5 & 5&13&8&8&0.09&0.02&34.9 \\
	CCZ$_{\# 3}$ &5 & 7&19&12&12&0.04&0.01&2219.68 \\
	CCZ$_{\# 2}$ &7 & 5&13&8&8&0.08&0.02&214.09 \\
	CCZ$_{\# 3}$ &7 & 7&19&12&12&0.04&0.01&23950.64 \\
	CCZ$_{\# 2}$ &11 & 5&13&8&8&0.08&0.02&7976.43 \\
	Random &5 & 3&8.26&11.38&4.52&$<$0.01&$<$0.01&1.40 \\
	Random &5 & 4&16.93&28.86&7.21&$<$0.01&$<$0.01&825.4959 \\
	Random &7 & 3&8.68&11.59&4.38&$<$0.01&$<$0.01&11.76 \\
	Random &7 & 4&16.88&28.75&7.13&$<$0.01&$<$0.01&10416.00\\
	Random &11 & 3&9.08&12.05&4.38&$<$0.01&$<$0.01&185.46 \\ \hline
		\end{tabularx}
	\end{table}

\section{Conclusions and Acknowledgements}

In this work we have generalised the phase polynomial type optimizers to qudit based quantum computers and have used it to demonstrate cost savings only possible in the qudit picture.
This motivates serious discussion into fundamental questions regarding the nature of  first generation fault-tolerant architectures, namely whether they use qubits or qudits. 

We acknowledge support by the Engineering and Physical Sciences Research Council (EPSRC) through grant EP/M024261/1.  We thank Mark Howard for discussions throughout the project.

\begin{appendices}

\section{Proof of Lemma~\ref{lem_imp}}
\label{ap_proof}
Let $U_f\in\langle \mathcal{G} \rangle$ be a unitary and $(A, \lambda)$ be an implementation for $f$ with $m$ columns.
We can efficiently generate a circuit, $C$, on $\mathcal{G}$ that implements $U_f$ from $(A, \lambda)$ using $m$ $M$ gates with the following algorithm:
\begin{enumerate}
\item Initialize an empty circuit, $C$.
\item For each $j \in [1,m]$:
	\begin{enumerate}
		\item Initialize an empty circuit, $D$.
		\item Let $H := \{i \ \mid A_{i,j} \neq 0\}$.
		\item Arbitrarily choose a $t \in H$.
		\item \label{step_lin1} Append $P_{A_{t,j}}$ on qudit line $t$ to $D$.
		\item For each $c \in H \setminus \{t\}$:
		\begin{enumerate}
			\item Append $P_{A_{c,j}}$ on qudit line $c$  to $D$.
			\item Append $SUM_{c,t}$ to $D$.
			\item Append $P_{\frac{1}{A_{c,j}}}$ on qudit line $c$ to $D$.
		\end{enumerate}
		\item \label{step_lin2}  Append $P_{\frac{1}{A_{t,j}}}$ on qudit line $t$ to $D$.
		\item Append $D$ to $C$.
		\item \label{step_m} Append $M^{\lambda_j}$ on qudit line $t$ to $C$.
		\item \label{step_uncomput} Append $D^\dagger$ to $C$.
	\end{enumerate}
\end{enumerate}
First, observe steps~\ref{step_lin1} to~\ref{step_lin2}, which creates a subcircuit $D$ using only $P_l$ and $SUM$ gates that maps the state of the \numth{t} qudit to a linear function of the $n$ input qudits $x_1, x_2, \dots, x_n$ that has coefficients given by the \numth{j} column of $A$.
After $D$ is appended to the output circuit, $C$, step~\ref{step_m} applies an $M^k$ gate with $k=\lambda_j$, which adds a term to the phase polynomial $f(\mathbf{x})$ equal to the aforementioned linear function cubed multiplied by $\lambda_j$, as required.
Finally, in step~\ref{step_uncomput}, the linear function is uncomputed by $D^\dagger$. The whole process is repeated for each of the $m$ columns of $A$.
Each iteration requires only one $M^k$ gate so the total number of $M$ gates required is $m$. The algorithm executes in $O(mn)$ steps and so is efficient.
We end this appendix with the disclamation that the above algorithm is not intended to be optimal with respect to the number of Clifford gates (i.e. $P_l$ and $SUM$) used.

\section{The Duplicate And Merge Optimizer}
\label{ap_dam}

The following is a direct generalisation of the TODD compiler from reference~\cite{Heyfron_2019} to qudits.
The key difference is that the the null space step from TODD is replaced with a multivariate cubic system, for which a common root must be found.
We refer the reader to Section 3.4 and Algorithm 1 of~\cite{Heyfron_2019} for an overview of the TODD compiler, which may aid in understanding DAM.
	
	\begin{definition}
		\label{def_duptra}
		\textbf{Duplication Transformation. } Let $A\in \mathbb{Z}_d^{(n,m)}$ be an implementation and $\mathbf{y} \in \mathbb{Z}_d^m$ be a vector. We define the \emph{duplication transformation} as follows:
		\begin{equation}
		A \rightarrow A + (\mathbf{c}_{b} - \mathbf{c}_{a})\mathbf{y}^T,
		\end{equation}
		where $\mathbf{c}_j$ is the $j$\textsuperscript{th} column of $A$.
	\end{definition}
	We can use this transformation to `create' duplicates as the following lemma shows.
	\begin{lemma}
		\label{lem_duptra}
		Let $A^\prime = A + (\mathbf{c}_{b} - \mathbf{c}_{a})\mathbf{y}^T$ and assume that  $y_a - y_b = 1$.
		It follows that $\mathbf{c}^\prime_a=\mathbf{c}^\prime_b$.
	\end{lemma}
	\begin{proof}
		From the definition of $A^\prime$,
		\begin{equation}
		A^\prime_{i,j} = A_{i,j} + z_iy_j,
		\end{equation}
		now substitute in $z_i \equiv A_{i,b} - A_{i,a}$,
		\begin{equation}
		\label{eq_l2_1}
		A^\prime_{i,j} = A_{i,j} + (A_{i,b} - A_{i,a})y_j.
		\end{equation}
		Apply equation \eqref{eq_l2_1} to both $(\mathbf{c}_a)_i \equiv A^\prime_{i,a}$ and $(\mathbf{c}_b)_i \equiv A^\prime_{i,b}$,
		\begin{equation}
		A^\prime_{i,b} = A_{i,b} + (A_{i,b} - A_{i,a})y_b,
		\end{equation}
		\begin{equation}
		\label{eq_l2_2}
		A^\prime_{i,a} = A_{i,a} + (A_{i,b} - A_{i,a})y_a.
		\end{equation}
		Substitute $y_a = y_b + 1$ into equation \eqref{eq_l2_2} and rearrange,
		\begin{align}
		A^\prime_{i,a} &= A_{i,a} + (A_{i,b} - A_{i,a})(y_b + 1) \\
		&= A_{i,a} + (A_{i,b} - A_{i,a})y_b + A_{i,b} - A_{i,a} \\
		&= A_{i,b} + (A_{i,b} - A_{i,a})y_b = A^\prime_{i,b}.
		\end{align}
		$\forall \ i \in [1,n]$ so $\mathbf{c}^\prime_{a} = \mathbf{c}^\prime_{b}$.
	\end{proof}
	
	The duplication transformation must not alter $f$.
	This leads to the condition $S^\prime_{\alpha,\beta,\gamma} = S_{\alpha,\beta,\gamma} \ \forall \ \alpha,\beta,\gamma \in [1,n]$, where $S^\prime$ and $S$ are the signature tensors for $A^\prime$ and $A$, respectively.
	So,
	\begin{align}
	S^\prime_{\alpha,\beta,\gamma} &= \sum_{j=1}^m \lambda_j A^\prime_{\alpha,j} A^\prime_{\beta,j} A^\prime_{\gamma,j}, \\
	&=  \sum_{j=1}^m \lambda_j (A_{\beta,j} + z_\beta y_j) (A_{\beta,j} + z_\beta y_j) (A_{\gamma,j} + z_\gamma y_j), \\
	&= \sum_{j=1}^m \lambda_j A_{\alpha,j} A_{\beta,j} A_{\gamma,j} + \Delta_{\alpha,\beta,\gamma} = S_{\alpha,\beta,\gamma} + \Delta_{\alpha,\beta,\gamma},
	\end{align}
	where we define,	
	\begin{equation}
	\begin{split}
	\label{eq_delta_1}
	\Delta_{\alpha,\beta,\gamma} :=& \sum_{j=1}^m \lambda_j (A_{\alpha,j} A_{\beta,j} z_\gamma y_j + A_{\beta,j} A_{\gamma,j} z_\alpha y_j + A_{\gamma,j} A_{\alpha,j} z_\beta y_j \\
	 &+ A_{\alpha,j}z_\beta z_\gamma y_j^2 + A_{\beta,j}z_\gamma z_\alpha y_j^2 + A_{\gamma,j}z_\alpha z_\beta y_j^2
	 + z_\alpha z_\beta z_\gamma y_j^3),
	\end{split}
	\end{equation}
	and $\mathbf{z} := \mathbf{c}_b - \mathbf{c}_a$.
	In order for $S^\prime = S$, we require that
	\begin{equation}
	\label{eq_sig_cons}
	\Delta_{\alpha,\beta,\gamma}=0 \quad \forall \quad \alpha,\beta,\gamma \in [1,n].
	\end{equation}
	This leads to a system of $\sum_{i=1}^3\binom{n}{i}$ cubic polynomials on $r$ variables ($y_1, y_2, \dots, y_m$) that can be rewritten as follows:
	\begin{align}
	\label{eq_dodd1}
	\sum_{j=1}^{m}\left(l_{\alpha,\beta,\gamma,j}y_j + q_{\alpha,\beta,\gamma,j}y_j^2 + c_{\alpha,\beta,\gamma,j}y_j^3\right) &= 0, \\
	\label{eq_dodd2} y_a - y_b - 1 = 0
	\end{align}
	where the linear, quadratic and cubic coefficients for variable $t$ are given by,
	\begin{align}
	l_{\alpha,\beta,\gamma,j} &= \lambda_j (A_{\alpha,j} A_{\beta,j} z_\gamma + A_{\beta,j} A_{\gamma,j} z_\alpha + A_{\gamma,j} A_{\alpha,j} z_\beta) \\
	q_{\alpha,\beta,\gamma,j} &= \lambda_j (A_{\alpha,j} z_\beta z_\gamma + A_{\beta,j} z_\gamma z_\alpha + A_{\gamma,j} z_\alpha z_\beta)\\
	c_{\alpha,\beta,\gamma,j} &= \lambda_j z_\alpha z_\beta z_\gamma,
	\end{align}
	respectively.
	
	Any $\mathbf{y}$ that is a simultaneous solution to equations \eqref{eq_dodd1} and equation \eqref{eq_dodd2} allows us to reduce the number of columns of $A$ using the duplication transformation from definition \eqref{def_duptra}. Unfortunately, the problem of solving a general multivariate cubic system such as this is known to be NP-complete.
	A brute-force solver that searches through every possible $\mathbf{y}$ runs in $O(d^m)$ time. However, we can significantly speed up the search using the following relinearisation technique. First, we introduce new variables, $y_{m+1}, y_{m+2}, \dots, y_{3m}$, such that
	\begin{align}
 	\label{eq_relin2}
 	y_{m+j} &= y_j^2, \\
 	y_{2m+j} &= y_j^3, \\
 	\end{align}
 	for all $j \in [1,m]$.
 	The system of equations from~\eqref{eq_dodd1} becomes:
 	\begin{equation}
 	\label{eq_relin}
 	\sum_{j=1}^m l_{\alpha,\beta,\gamma,j}y_j + \sum_{k=m+1}^{2m} q_{\alpha,\beta,\gamma,k}y_{k} + \sum_{l=2m+1}^{3m} c_{\alpha,\beta,\gamma,l}y_{l} = 0,
 	\end{equation}
 	which is linear in the $\{y_j\}$. Let $D$ be the coefficient matrix defined as follows:
 	For each triple $\{(\alpha, \beta, \gamma) \ \mid \ \alpha \leq \beta \leq \gamma \in [1,n]$, there exists a row in $D$ of the following form:
	\begin{equation}
	\textsc{Row}_{\alpha,\beta,\gamma}(D) = \begin{pmatrix} (l_{\alpha,\beta,\gamma,j}) & (q_{\alpha,\beta,\gamma,j}) & (c_{\alpha,\beta,\gamma,j}) \end{pmatrix}.
	\end{equation}
 	Now we can calculate a complete basis for the solutions to equation~\eqref{eq_relin} by calculating the right null space of $D$, which we denote $N_D$. We can think of the columns of $N_D$ as a basis for the `partial' solutions of the system of equations~\eqref{eq_dodd1}. In order to promote them to `full' solutions, we need to enforce conditions from equations~\eqref{eq_relin2} and~\eqref{eq_dodd2}, which we do using the following algorithm. 
 	\begin{enumerate}
	\item Form $N_D^\prime$ by erasing all but the first $m$ rows of $N_D$.
	\item Form $N_D^{\prime\prime}$ by column-reducing $N_D^\prime$ and subsequently removing every all-zero column.
	\item Let $\mu := \textsc{Cols}(N_D^{\prime\prime})$.
	\item For each $\mathbf{x} \in \mathbb{Z}_d^{\mu}$:
	\begin{enumerate}
	\item Construct  $\mathbf{y}_{\mathbf{x}} = N_D^{\prime\prime}\mathbf{x}$
	\item Construct $\mathbf{y}_{\mathbf{x}}^\prime = \begin{pmatrix} \mathbf{y}_{\mathbf{x}} \\ \mathbf{y}_{\mathbf{x}}^2 \\ \mathbf{y}_{\mathbf{x}}^3\end{pmatrix}$
	\item If $D\mathbf{y}_{\mathbf{x}}^\prime = \mathbf{0}$ and~\eqref{eq_dodd2} holds, then return $\mathbf{y}_{\mathbf{x}}$.
	\end{enumerate}
	\item Return ``No Solution''.
	\end{enumerate}
	
	In effect, the relinearlisation method replaces a search over every $\mathbf{y}\in \mathbb{Z}_d^m$ with a search over every $\mathbf{x}\in\mathbb{Z}_d^\mu$, so it only runs faster if $\mu < m$.
	It is certainly the case that $\mu \leq m$ as $\mu$ is the column rank of $N_D^{\prime}$, which has $m$ rows.
	Whether or not the strict inequality holds depends on the input circuit but in practice, we often find that it does hold and leads to significant speed-up over the naive brute force approach.

\end{appendices}


\begin{thebibliography}{31}%
	\makeatletter
	\providecommand \@ifxundefined [1]{%
		\@ifx{#1\undefined}
	}%
	\providecommand \@ifnum [1]{%
		\ifnum #1\expandafter \@firstoftwo
		\else \expandafter \@secondoftwo
		\fi
	}%
	\providecommand \@ifx [1]{%
		\ifx #1\expandafter \@firstoftwo
		\else \expandafter \@secondoftwo
		\fi
	}%
	\providecommand \natexlab [1]{#1}%
	\providecommand \enquote  [1]{``#1''}%
	\providecommand \bibnamefont  [1]{#1}%
	\providecommand \bibfnamefont [1]{#1}%
	\providecommand \citenamefont [1]{#1}%
	\providecommand \href@noop [0]{\@secondoftwo}%
	\providecommand \href [0]{\begingroup \@sanitize@url \@href}%
	\providecommand \@href[1]{\@@startlink{#1}\@@href}%
	\providecommand \@@href[1]{\endgroup#1\@@endlink}%
	\providecommand \@sanitize@url [0]{\catcode `\\12\catcode `\$12\catcode
		`\&12\catcode `\#12\catcode `\^12\catcode `\_12\catcode `\%12\relax}%
	\providecommand \@@startlink[1]{}%
	\providecommand \@@endlink[0]{}%
	\providecommand \url  [0]{\begingroup\@sanitize@url \@url }%
	\providecommand \@url [1]{\endgroup\@href {#1}{\urlprefix }}%
	\providecommand \urlprefix  [0]{URL }%
	\providecommand \Eprint [0]{\href }%
	\providecommand \doibase [0]{http://dx.doi.org/}%
	\providecommand \selectlanguage [0]{\@gobble}%
	\providecommand \bibinfo  [0]{\@secondoftwo}%
	\providecommand \bibfield  [0]{\@secondoftwo}%
	\providecommand \translation [1]{[#1]}%
	\providecommand \BibitemOpen [0]{}%
	\providecommand \bibitemStop [0]{}%
	\providecommand \bibitemNoStop [0]{.\EOS\space}%
	\providecommand \EOS [0]{\spacefactor3000\relax}%
	\providecommand \BibitemShut  [1]{\csname bibitem#1\endcsname}%
	\let\auto@bib@innerbib\@empty
	\bibitem [{\citenamefont {Brusentsov}\ and\ \citenamefont
		{Ramil~Alvarez}(2011)}]{Brusentsov_2011}%
	\BibitemOpen
	\bibfield  {author} {\bibinfo {author} {\bibfnamefont {N.~P.}\ \bibnamefont
			{Brusentsov}}\ and\ \bibinfo {author} {\bibfnamefont {J.}~\bibnamefont
			{Ramil~Alvarez}},\ }\href {\doibase 10.1007/978-3-642-22816-2_10} {\bibfield
		{journal} {\bibinfo  {journal} {IFIP Advances in Information and
				Communication Technology}\ ,\ \bibinfo {pages} {74}} (\bibinfo {year}
		{2011})}\BibitemShut {NoStop}%
	\bibitem [{\citenamefont {Duclos-Cianci}\ and\ \citenamefont
		{Poulin}(2013)}]{Duclos-Cianci_2013}%
	\BibitemOpen
	\bibfield  {author} {\bibinfo {author} {\bibfnamefont {G.}~\bibnamefont
			{Duclos-Cianci}}\ and\ \bibinfo {author} {\bibfnamefont {D.}~\bibnamefont
			{Poulin}},\ }\href {\doibase 10.1103/PhysRevA.87.062338} {\bibfield
		{journal} {\bibinfo  {journal} {Phys. Rev. A}\ }\textbf {\bibinfo {volume}
			{87}},\ \bibinfo {pages} {062338} (\bibinfo {year} {2013})}\BibitemShut
	{NoStop}%
	\bibitem [{\citenamefont {Anwar}\ \emph {et~al.}(2014)\citenamefont {Anwar},
		\citenamefont {Brown}, \citenamefont {Campbell},\ and\ \citenamefont
		{Browne}}]{Anwar_2014}%
	\BibitemOpen
	\bibfield  {author} {\bibinfo {author} {\bibfnamefont {H.}~\bibnamefont
			{Anwar}}, \bibinfo {author} {\bibfnamefont {B.~J.}\ \bibnamefont {Brown}},
		\bibinfo {author} {\bibfnamefont {E.~T.}\ \bibnamefont {Campbell}}, \ and\
		\bibinfo {author} {\bibfnamefont {D.~E.}\ \bibnamefont {Browne}},\ }\href
	{http://stacks.iop.org/1367-2630/16/i=6/a=063038} {\bibfield  {journal}
		{\bibinfo  {journal} {New Journal of Physics}\ }\textbf {\bibinfo {volume}
			{16}},\ \bibinfo {pages} {063038} (\bibinfo {year} {2014})}\BibitemShut
	{NoStop}%
	\bibitem [{\citenamefont {Hutter}\ \emph {et~al.}(2015)\citenamefont {Hutter},
		\citenamefont {Loss},\ and\ \citenamefont {Wootton}}]{Hutter_2015}%
	\BibitemOpen
	\bibfield  {author} {\bibinfo {author} {\bibfnamefont {A.}~\bibnamefont
			{Hutter}}, \bibinfo {author} {\bibfnamefont {D.}~\bibnamefont {Loss}}, \ and\
		\bibinfo {author} {\bibfnamefont {J.~R.}\ \bibnamefont {Wootton}},\ }\href
	{http://stacks.iop.org/1367-2630/17/i=3/a=035017} {\bibfield  {journal}
		{\bibinfo  {journal} {New Journal of Physics}\ }\textbf {\bibinfo {volume}
			{17}},\ \bibinfo {pages} {035017} (\bibinfo {year} {2015})}\BibitemShut
	{NoStop}%
	\bibitem [{\citenamefont {Watson}\ \emph
		{et~al.}(2015{\natexlab{a}})\citenamefont {Watson}, \citenamefont {Campbell},
		\citenamefont {Anwar},\ and\ \citenamefont {Browne}}]{Watson_2015_a}%
	\BibitemOpen
	\bibfield  {author} {\bibinfo {author} {\bibfnamefont {F.~H.~E.}\
			\bibnamefont {Watson}}, \bibinfo {author} {\bibfnamefont {E.~T.}\
			\bibnamefont {Campbell}}, \bibinfo {author} {\bibfnamefont {H.}~\bibnamefont
			{Anwar}}, \ and\ \bibinfo {author} {\bibfnamefont {D.~E.}\ \bibnamefont
			{Browne}},\ }\href {\doibase 10.1103/PhysRevA.92.022312} {\bibfield
		{journal} {\bibinfo  {journal} {Phys. Rev. A}\ }\textbf {\bibinfo {volume}
			{92}},\ \bibinfo {pages} {022312} (\bibinfo {year}
		{2015}{\natexlab{a}})}\BibitemShut {NoStop}%
	\bibitem [{\citenamefont {Watson}\ \emph
		{et~al.}(2015{\natexlab{b}})\citenamefont {Watson}, \citenamefont {Anwar},\
		and\ \citenamefont {Browne}}]{Watson_2015_b}%
	\BibitemOpen
	\bibfield  {author} {\bibinfo {author} {\bibfnamefont {F.~H.~E.}\
			\bibnamefont {Watson}}, \bibinfo {author} {\bibfnamefont {H.}~\bibnamefont
			{Anwar}}, \ and\ \bibinfo {author} {\bibfnamefont {D.~E.}\ \bibnamefont
			{Browne}},\ }\href {\doibase 10.1103/PhysRevA.92.032309} {\bibfield
		{journal} {\bibinfo  {journal} {Phys. Rev. A}\ }\textbf {\bibinfo {volume}
			{92}},\ \bibinfo {pages} {032309} (\bibinfo {year}
		{2015}{\natexlab{b}})}\BibitemShut {NoStop}%
	\bibitem [{\citenamefont {Anwar}\ \emph {et~al.}(2012)\citenamefont {Anwar},
		\citenamefont {Campbell},\ and\ \citenamefont {Browne}}]{anwar2012qutrit}%
	\BibitemOpen
	\bibfield  {author} {\bibinfo {author} {\bibfnamefont {H.}~\bibnamefont
			{Anwar}}, \bibinfo {author} {\bibfnamefont {E.~T.}\ \bibnamefont {Campbell}},
		\ and\ \bibinfo {author} {\bibfnamefont {D.~E.}\ \bibnamefont {Browne}},\
	}\href@noop {} {\bibfield  {journal} {\bibinfo  {journal} {New Journal of
				Physics}\ }\textbf {\bibinfo {volume} {14}},\ \bibinfo {pages} {063006}
		(\bibinfo {year} {2012})}\BibitemShut {NoStop}%
	\bibitem [{\citenamefont {Campbell}\ \emph {et~al.}(2012)\citenamefont
		{Campbell}, \citenamefont {Anwar},\ and\ \citenamefont
		{Browne}}]{Campbell_2012}%
	\BibitemOpen
	\bibfield  {author} {\bibinfo {author} {\bibfnamefont {E.~T.}\ \bibnamefont
			{Campbell}}, \bibinfo {author} {\bibfnamefont {H.}~\bibnamefont {Anwar}}, \
		and\ \bibinfo {author} {\bibfnamefont {D.~E.}\ \bibnamefont {Browne}},\
	}\href {\doibase 10.1103/PhysRevX.2.041021} {\bibfield  {journal} {\bibinfo
			{journal} {Phys. Rev. X}\ }\textbf {\bibinfo {volume} {2}},\ \bibinfo {pages}
		{041021} (\bibinfo {year} {2012})}\BibitemShut {NoStop}%
	\bibitem [{\citenamefont {Campbell}(2014)}]{Campbell_2014}%
	\BibitemOpen
	\bibfield  {author} {\bibinfo {author} {\bibfnamefont {E.~T.}\ \bibnamefont
			{Campbell}},\ }\href {\doibase 10.1103/PhysRevLett.113.230501} {\bibfield
		{journal} {\bibinfo  {journal} {Phys. Rev. Lett.}\ }\textbf {\bibinfo
			{volume} {113}},\ \bibinfo {pages} {230501} (\bibinfo {year}
		{2014})}\BibitemShut {NoStop}%
	\bibitem [{\citenamefont {Haah}\ \emph {et~al.}(2017)\citenamefont {Haah},
		\citenamefont {Hastings}, \citenamefont {Poulin},\ and\ \citenamefont
		{Wecker}}]{haah2017magic}%
	\BibitemOpen
	\bibfield  {author} {\bibinfo {author} {\bibfnamefont {J.}~\bibnamefont
			{Haah}}, \bibinfo {author} {\bibfnamefont {M.~B.}\ \bibnamefont {Hastings}},
		\bibinfo {author} {\bibfnamefont {D.}~\bibnamefont {Poulin}}, \ and\ \bibinfo
		{author} {\bibfnamefont {D.}~\bibnamefont {Wecker}},\ }\href@noop {}
	{\bibfield  {journal} {\bibinfo  {journal} {Quantum}\ }\textbf {\bibinfo
			{volume} {1}},\ \bibinfo {pages} {31} (\bibinfo {year} {2017})}\BibitemShut
	{NoStop}%
	\bibitem [{\citenamefont {Krishna}\ and\ \citenamefont
		{Tillich}(2018)}]{krishna2018towards}%
	\BibitemOpen
	\bibfield  {author} {\bibinfo {author} {\bibfnamefont {A.}~\bibnamefont
			{Krishna}}\ and\ \bibinfo {author} {\bibfnamefont {J.-P.}\ \bibnamefont
			{Tillich}},\ }\href@noop {} {\bibfield  {journal} {\bibinfo  {journal} {arXiv
				preprint arXiv:1811.08461}\ } (\bibinfo {year} {2018})}\BibitemShut {NoStop}%
	\bibitem [{\citenamefont {Khan}\ and\ \citenamefont
		{Perkowski}(2005)}]{khan2005synthesis}%
	\BibitemOpen
	\bibfield  {author} {\bibinfo {author} {\bibfnamefont {F.~S.}\ \bibnamefont
			{Khan}}\ and\ \bibinfo {author} {\bibfnamefont {M.}~\bibnamefont
			{Perkowski}},\ }\href@noop {} {\bibfield  {journal} {\bibinfo  {journal}
			{arXiv preprint quant-ph/0511041}\ } (\bibinfo {year} {2005})}\BibitemShut
	{NoStop}%
	\bibitem [{\citenamefont {Bocharov}\ \emph {et~al.}(2017)\citenamefont
		{Bocharov}, \citenamefont {Roetteler},\ and\ \citenamefont
		{Svore}}]{bocharov2017factoring}%
	\BibitemOpen
	\bibfield  {author} {\bibinfo {author} {\bibfnamefont {A.}~\bibnamefont
			{Bocharov}}, \bibinfo {author} {\bibfnamefont {M.}~\bibnamefont {Roetteler}},
		\ and\ \bibinfo {author} {\bibfnamefont {K.~M.}\ \bibnamefont {Svore}},\
	}\href@noop {} {\bibfield  {journal} {\bibinfo  {journal} {Physical Review
				A}\ }\textbf {\bibinfo {volume} {96}},\ \bibinfo {pages} {012306} (\bibinfo
		{year} {2017})}\BibitemShut {NoStop}%
	\bibitem [{\citenamefont {Matsumoto}\ and\ \citenamefont
		{Amano}(2008)}]{Matsumoto_2008}%
	\BibitemOpen
	\bibfield  {author} {\bibinfo {author} {\bibfnamefont {K.}~\bibnamefont
			{Matsumoto}}\ and\ \bibinfo {author} {\bibfnamefont {K.}~\bibnamefont
			{Amano}},\ }\href {https://arxiv.org/abs/0806.3834} {\bibfield  {journal}
		{\bibinfo  {journal} {Pre-print arXiv:0806.3834}\ } (\bibinfo {year}
		{2008})}\BibitemShut {NoStop}%
	\bibitem [{\citenamefont {Giles}\ and\ \citenamefont
		{Selinger}(2013)}]{Giles_2013}%
	\BibitemOpen
	\bibfield  {author} {\bibinfo {author} {\bibfnamefont {B.}~\bibnamefont
			{Giles}}\ and\ \bibinfo {author} {\bibfnamefont {P.}~\bibnamefont
			{Selinger}},\ }\href {https://arxiv.org/abs/1312.6584} {\bibfield  {journal}
		{\bibinfo  {journal} {Pre-print arXiv:1312.6584}\ } (\bibinfo {year}
		{2013})}\BibitemShut {NoStop}%
	\bibitem [{\citenamefont {Kliuchnikov}\ \emph {et~al.}(2013)\citenamefont
		{Kliuchnikov}, \citenamefont {Maslov},\ and\ \citenamefont
		{Mosca}}]{Kliuchnikov_2013}%
	\BibitemOpen
	\bibfield  {author} {\bibinfo {author} {\bibfnamefont {V.}~\bibnamefont
			{Kliuchnikov}}, \bibinfo {author} {\bibfnamefont {D.}~\bibnamefont {Maslov}},
		\ and\ \bibinfo {author} {\bibfnamefont {M.}~\bibnamefont {Mosca}},\ }\href
	{http://dl.acm.org/citation.cfm?id=2535649.2535653} {\bibfield  {journal}
		{\bibinfo  {journal} {Quantum Info. Comput.}\ }\textbf {\bibinfo {volume}
			{13}},\ \bibinfo {pages} {607} (\bibinfo {year} {2013})}\BibitemShut
	{NoStop}%
	\bibitem [{\citenamefont {Gosset}\ \emph {et~al.}(2014)\citenamefont {Gosset},
		\citenamefont {Kliuchnikov}, \citenamefont {Mosca},\ and\ \citenamefont
		{Russo}}]{Gosset_2014}%
	\BibitemOpen
	\bibfield  {author} {\bibinfo {author} {\bibfnamefont {D.}~\bibnamefont
			{Gosset}}, \bibinfo {author} {\bibfnamefont {V.}~\bibnamefont {Kliuchnikov}},
		\bibinfo {author} {\bibfnamefont {M.}~\bibnamefont {Mosca}}, \ and\ \bibinfo
		{author} {\bibfnamefont {V.}~\bibnamefont {Russo}},\ }\href
	{https://www.microsoft.com/en-us/research/publication/an-algorithm-for-the-t-count/}
	{\bibfield  {journal} {\bibinfo  {journal} {Quantum Info. Comput.}\ }\textbf
		{\bibinfo {volume} {14}},\ \bibinfo {pages} {1261} (\bibinfo {year}
		{2014})}\BibitemShut {NoStop}%
	\bibitem [{\citenamefont {Amy}\ \emph {et~al.}(2014)\citenamefont {Amy},
		\citenamefont {Maslov},\ and\ \citenamefont {Mosca}}]{Amy_2014}%
	\BibitemOpen
	\bibfield  {author} {\bibinfo {author} {\bibfnamefont {M.}~\bibnamefont
			{Amy}}, \bibinfo {author} {\bibfnamefont {D.}~\bibnamefont {Maslov}}, \ and\
		\bibinfo {author} {\bibfnamefont {M.}~\bibnamefont {Mosca}},\ }\href
	{\doibase 10.1109/TCAD.2014.2341953} {\bibfield  {journal} {\bibinfo
			{journal} {IEEE Transactions on Computer-Aided Design of Integrated Circuits
				and Systems}\ }\textbf {\bibinfo {volume} {33}},\ \bibinfo {pages} {1476}
		(\bibinfo {year} {2014})}\BibitemShut {NoStop}%
	\bibitem [{\citenamefont {Amy}\ and\ \citenamefont {Mosca}(2016)}]{Amy_2016}%
	\BibitemOpen
	\bibfield  {author} {\bibinfo {author} {\bibfnamefont {M.}~\bibnamefont
			{Amy}}\ and\ \bibinfo {author} {\bibfnamefont {M.}~\bibnamefont {Mosca}},\
	}\href {https://arxiv.org/abs/1601.07363} {\bibfield  {journal} {\bibinfo
			{journal} {Pre-print arXiv:1601.07363}\ } (\bibinfo {year}
		{2016})}\BibitemShut {NoStop}%
	\bibitem [{\citenamefont {Campbell}\ and\ \citenamefont
		{Howard}(2017)}]{Campbell_2017}%
	\BibitemOpen
	\bibfield  {author} {\bibinfo {author} {\bibfnamefont {E.~T.}\ \bibnamefont
			{Campbell}}\ and\ \bibinfo {author} {\bibfnamefont {M.}~\bibnamefont
			{Howard}},\ }\href {\doibase 10.1103/PhysRevA.95.022316} {\bibfield
		{journal} {\bibinfo  {journal} {Phys. Rev. A}\ }\textbf {\bibinfo {volume}
			{95}},\ \bibinfo {pages} {022316} (\bibinfo {year} {2017})}\BibitemShut
	{NoStop}%
	\bibitem [{\citenamefont {Heyfron}\ and\ \citenamefont
		{Campbell}(2019)}]{Heyfron_2019}%
	\BibitemOpen
	\bibfield  {author} {\bibinfo {author} {\bibfnamefont {L.~E.}\ \bibnamefont
			{Heyfron}}\ and\ \bibinfo {author} {\bibfnamefont {E.~T.}\ \bibnamefont
			{Campbell}},\ }\href {http://stacks.iop.org/2058-9565/4/i=1/a=015004}
	{\bibfield  {journal} {\bibinfo  {journal} {Quantum Science and Technology}\
		}\textbf {\bibinfo {volume} {4}},\ \bibinfo {pages} {015004} (\bibinfo {year}
		{2019})}\BibitemShut {NoStop}%
	\bibitem [{\citenamefont {Nam}\ \emph {et~al.}(2017)\citenamefont {Nam},
		\citenamefont {J.~Ross}, \citenamefont {Su}, \citenamefont {Childs},\ and\
		\citenamefont {Maslov}}]{Nam_2017}%
	\BibitemOpen
	\bibfield  {author} {\bibinfo {author} {\bibfnamefont {Y.}~\bibnamefont
			{Nam}}, \bibinfo {author} {\bibfnamefont {N.}~\bibnamefont {J.~Ross}},
		\bibinfo {author} {\bibfnamefont {Y.}~\bibnamefont {Su}}, \bibinfo {author}
		{\bibfnamefont {A.}~\bibnamefont {Childs}}, \ and\ \bibinfo {author}
		{\bibfnamefont {D.}~\bibnamefont {Maslov}},\ }\href {\doibase
		10.1038/s41534-018-0072-4} {\bibfield  {journal} {\bibinfo  {journal} {npj
				Quantum Information}\ }\textbf {\bibinfo {volume} {4}} (\bibinfo {year}
		{2017}),\ 10.1038/s41534-018-0072-4}\BibitemShut {NoStop}%
	\bibitem [{\citenamefont {Glaudell}\ \emph {et~al.}(2018)\citenamefont
		{Glaudell}, \citenamefont {Ross},\ and\ \citenamefont
		{Taylor}}]{Glaudell_2018}%
	\BibitemOpen
	\bibfield  {author} {\bibinfo {author} {\bibfnamefont {A.~N.}\ \bibnamefont
			{Glaudell}}, \bibinfo {author} {\bibfnamefont {N.~J.}\ \bibnamefont {Ross}},
		\ and\ \bibinfo {author} {\bibfnamefont {J.~M.}\ \bibnamefont {Taylor}},\
	}\href {https://arxiv.org/abs/1803.05047} {\bibfield  {journal} {\bibinfo
			{journal} {Pre-print arXiv:1803.05047}\ } (\bibinfo {year}
		{2018})}\BibitemShut {NoStop}%
	\bibitem [{\citenamefont {Jones}(2013)}]{Jones_2013}%
	\BibitemOpen
	\bibfield  {author} {\bibinfo {author} {\bibfnamefont {C.}~\bibnamefont
			{Jones}},\ }\href@noop {} {\bibfield  {journal} {\bibinfo  {journal}
			{Physical Review A}\ }\textbf {\bibinfo {volume} {87}},\ \bibinfo {pages}
		{022328} (\bibinfo {year} {2013})}\BibitemShut {NoStop}%
	\bibitem [{\citenamefont {Howard}\ and\ \citenamefont
		{Vala}(2012)}]{Howard_2012}%
	\BibitemOpen
	\bibfield  {author} {\bibinfo {author} {\bibfnamefont {M.}~\bibnamefont
			{Howard}}\ and\ \bibinfo {author} {\bibfnamefont {J.}~\bibnamefont {Vala}},\
	}\href@noop {} {\bibfield  {journal} {\bibinfo  {journal} {Physical Review
				A}\ }\textbf {\bibinfo {volume} {86}},\ \bibinfo {pages} {022316} (\bibinfo
		{year} {2012})}\BibitemShut {NoStop}%
	\bibitem [{\citenamefont {Cui}\ \emph {et~al.}(2017)\citenamefont {Cui},
		\citenamefont {Gottesman},\ and\ \citenamefont {Krishna}}]{Cui_2017}%
	\BibitemOpen
	\bibfield  {author} {\bibinfo {author} {\bibfnamefont {S.~X.}\ \bibnamefont
			{Cui}}, \bibinfo {author} {\bibfnamefont {D.}~\bibnamefont {Gottesman}}, \
		and\ \bibinfo {author} {\bibfnamefont {A.}~\bibnamefont {Krishna}},\ }\href
	{\doibase 10.1103/PhysRevA.95.012329} {\bibfield  {journal} {\bibinfo
			{journal} {Phys. Rev. A}\ }\textbf {\bibinfo {volume} {95}},\ \bibinfo
		{pages} {012329} (\bibinfo {year} {2017})}\BibitemShut {NoStop}%
	\bibitem [{\citenamefont {Fowler}\ \emph {et~al.}(2012)\citenamefont {Fowler},
		\citenamefont {Whiteside},\ and\ \citenamefont
		{Hollenberg}}]{fowler2012towards}%
	\BibitemOpen
	\bibfield  {author} {\bibinfo {author} {\bibfnamefont {A.~G.}\ \bibnamefont
			{Fowler}}, \bibinfo {author} {\bibfnamefont {A.~C.}\ \bibnamefont
			{Whiteside}}, \ and\ \bibinfo {author} {\bibfnamefont {L.~C.}\ \bibnamefont
			{Hollenberg}},\ }\href@noop {} {\bibfield  {journal} {\bibinfo  {journal}
			{Physical review letters}\ }\textbf {\bibinfo {volume} {108}},\ \bibinfo
		{pages} {180501} (\bibinfo {year} {2012})}\BibitemShut {NoStop}%
	\bibitem [{\citenamefont {O'Gorman}\ and\ \citenamefont
		{Campbell}(2017)}]{o2017quantum}%
	\BibitemOpen
	\bibfield  {author} {\bibinfo {author} {\bibfnamefont {J.}~\bibnamefont
			{O'Gorman}}\ and\ \bibinfo {author} {\bibfnamefont {E.~T.}\ \bibnamefont
			{Campbell}},\ }\href@noop {} {\bibfield  {journal} {\bibinfo  {journal}
			{Physical Review A}\ }\textbf {\bibinfo {volume} {95}},\ \bibinfo {pages}
		{032338} (\bibinfo {year} {2017})}\BibitemShut {NoStop}%
	\bibitem [{\citenamefont {Babbush}\ \emph {et~al.}(2018)\citenamefont
		{Babbush}, \citenamefont {Gidney}, \citenamefont {Berry}, \citenamefont
		{Wiebe}, \citenamefont {McClean}, \citenamefont {Paler}, \citenamefont
		{Fowler},\ and\ \citenamefont {Neven}}]{babbush2018encoding}%
	\BibitemOpen
	\bibfield  {author} {\bibinfo {author} {\bibfnamefont {R.}~\bibnamefont
			{Babbush}}, \bibinfo {author} {\bibfnamefont {C.}~\bibnamefont {Gidney}},
		\bibinfo {author} {\bibfnamefont {D.~W.}\ \bibnamefont {Berry}}, \bibinfo
		{author} {\bibfnamefont {N.}~\bibnamefont {Wiebe}}, \bibinfo {author}
		{\bibfnamefont {J.}~\bibnamefont {McClean}}, \bibinfo {author} {\bibfnamefont
			{A.}~\bibnamefont {Paler}}, \bibinfo {author} {\bibfnamefont
			{A.}~\bibnamefont {Fowler}}, \ and\ \bibinfo {author} {\bibfnamefont
			{H.}~\bibnamefont {Neven}},\ }\href@noop {} {\bibfield  {journal} {\bibinfo
			{journal} {arXiv preprint arXiv:1805.03662}\ } (\bibinfo {year}
		{2018})}\BibitemShut {NoStop}%
	\bibitem [{Note1()}]{Note1}%
	\BibitemOpen
	\bibinfo {note} {Note that for notational convenience, we often write an
		implementation as a single matrix where $\lambda $ is the final row and the
		rest is the $A$ matrix with a separating horizontal line between
		them.}\BibitemShut {Stop}%
	\bibitem [{Note2()}]{Note2}%
	\BibitemOpen
	\bibinfo {note} {Execution times were obtained on a laptop with an Intel Core
		i7 2.40GHz processor, 12GB of RAM running Microsoft Windows 10 Home
		edition.}\BibitemShut {Stop}%
\end{thebibliography}

%

\end{document}